\documentclass[submission]{eptcs} 
\usepackage{breakurl}             
\usepackage[english]{babel}
\usepackage[utf8x]{inputenc}

\usepackage{amsmath,amsfonts,amsthm,latexsym,url}
\usepackage{dsfont}
\usepackage{mathptmx,graphicx,makeidx}
\usepackage[vlined,linesnumbered]{algorithm2e}
\usepackage{xcolor}

\newtheorem{property}{Property}
\newtheorem{definition}{Definition}
\newtheorem{proposition}{Proposition}
\newtheorem{corollary}{Corollary}
\newtheorem{lemma}{Lemma}

\providecommand{\R}{\ensuremath{{\mathds R}}}
\providecommand{\rgives}{\ensuremath{\stackrel{*}{\Longrightarrow}}}
\providecommand{\Grammarclass}{\ensuremath{\mathcal{G}}}
\providecommand{\PCFGclass}{\textup{\textbf{PCFG}(\ensuremath{\Sigma})}}

\providecommand{\Nrat}{\ensuremath{{\mathds N}}}
\providecommand{\Lone}{\textup{L1}}
\providecommand{\Ltwo}{\textup{L2}}
\providecommand{\Linf}{\textup{L}\ensuremath{_{\infty}}}

\providecommand{\CFG}{\ensuremath{\textsc{Cfg}}}
\providecommand{\PCFG}{\ensuremath{\textsc{Pcfg}}}
\providecommand{\A}{\ensuremath{{\cal{A}}}}

\providecommand{\AC}{\ensuremath{\textsc{Ac}}}
\providecommand{\TAC}{\ensuremath{\textsc{Tac}}}

\providecommand{\es}{\ensuremath{\lambda}}

\providecommand{\PCP}{\ensuremath{\textsc{Pcp}}}

\providecommand{\PFA}{\ensuremath{\textsc{Pfa}}}

\providecommand{\HMM}{\ensuremath{\textsc{Hmm}}}
\providecommand{\Naming}{\ensuremath{\mathbb L}}

\providecommand{\Prob}{\ensuremath{Pr}}

\providecommand{\len}{\textup{len}}
\providecommand{\Dis}{\ensuremath{{\cal{D}}}}

\providecommand{\Sigmastar}{\ensuremath{\mathop{\Sigma^{\star}}}}

\providecommand{\COEM}{\ensuremath{\textsc{CoEm}}}
\providecommand{\gives}{\ensuremath{\mathop{\rightarrow}}}
\title{On the Computation of Distances for Probabilistic Context-Free Grammars}
\author{Colin de la Higuera\qquad\qquad James Scicluna
\institute{LINA, UMR 6241\\
University of Nantes\thanks{The first and second author acknowledge partial support by the R\'{e}gion des Pays de la Loire.}\\
France}
\email{cdlh@univ-nantes.fr\qquad james.scicluna@univ-nantes.fr}
\and
Mark-Jan Nederhof
\institute{School of Computer Science \\
University of St Andrews \\ UK}
}
\def\titlerunning{On the Computation of Distances for Probabilistic Context-Free Grammars}
\def\authorrunning{C. de la Higuera, J. Scicluna \& M.-J. Nederhof }
\begin{document}
\maketitle

\begin{abstract}
Probabilistic context-free grammars ($\PCFG$s) are used to define distributions over strings, and are powerful modelling tools in a number of areas, including natural language processing, software engineering, model checking, bio-informatics, and pattern recognition. A common important question is that of comparing the distributions generated or modelled by these grammars: this is done through checking language equivalence and computing distances. Two $\PCFG$s are language equivalent if every string has identical probability with both grammars. This also means that the distance (whichever norm is used) is null. It is known that the language equivalence problem is interreducible with that  of multiple ambiguity for context-free grammars, a long-standing open question.
In this work, we prove that computing distances corresponds to solving undecidable questions: this is the case for the $\Lone$, $\Ltwo$ norm, the variation distance and the Kullback-Leibler divergence. Two more results are less negative: 1. The most probable string can be computed, and, 2. The Chebyshev distance (where the distance between two distributions is the maximum difference of probabilities over all strings) is interreducible with the language equivalence problem.
\end{abstract}
\section{General motivation and introduction}

Probabilistic context-free grammars ($\PCFG$s) are powerful modelling tools in a number of areas, including natural language processing, software engineering, model checking, bio-informatics, and pattern recognition. In natural language processing, these grammars are used as language models \cite{jurafsky95,benedi05} or for parsing natural language \cite{johnson98,klein03}. In model checking the crucial questions of program equivalence or meeting specifications will often be solved through tackling the grammar equivalence problem \cite{espa04,fore12}.

In pattern recognition, probabilistic context-free grammars have been proposed and used for 40 years \cite{fubo75}.
In bio-informatics structural dependencies are modelled through context-free rules, whose probabilities can be estimated \cite{saka94,salv02}.

In all these areas, the following questions are important: given two grammars, are they equivalent? Two grammars are equivalent (\emph{strongly}, or \emph{language equivalent}) when every string has identical probability in each distribution. More generally, a distance between distributions expresses how close they are, with a zero distance coinciding with equivalence.

Furthermore, in many areas, these probabilistic models are to be learnt. When learning, comparison between states or nonterminals often determines if a merge or a generalization is to take place.  Key grammatical inference operations \cite{higu10} will depend on the precision with which an equivalence is tested or a distance is measured.

In the case of probabilistic finite automata, these questions have been analysed with care. The initial study by Balasubramanian \cite{bala93} showed that the equivalence problem for hidden Markov models admitted a polynomial algorithm. Later, a number of authors showed that the Euclidian distance could be computed in polynomial time \cite{lyng99}.  This important result made use of the key concept of co-emission. A similar result was obtained for $\PFA$s \cite{murg04}. 

Negative results were also obtained: The $\Lone$ distance, and the variation distance were proved to be intractable \cite{lyng01}.
Relative entropy (or Kullback-Leibler divergence) cannot be computed between general $\PFA$s. But it can be computed for deterministic \cite{carr97} and unambiguous \cite{cort08} $\PFA$s.

A related problem, that of computing the most probable string (also called the \emph{consensus string}) was proved to be an NP-hard problem \cite{casa00a,lyng02};  heuristics and parameterized algorithms have been proposed \cite{higu13b}. Some of these results and techniques were also extended to $\PCFG$s \cite{higu13a}.

The co-emission of a $\PCFG$ and a hidden Markov model was discussed by \cite{jago01}, who formulated the problem in terms of finding the solution to a system of quadratic equations; for the difficulties in solving such systems, see also \cite{etessami2009}.
By a related construction, a probabilistic finite automaton can be obtained from a given unambiguous (non-probabilistic) finite automaton and a given $\PCFG$, in such a way that the Kullback-Leibler divergence between the two probability distributions is minimal \cite{nede04}.

The same problems have been far less studied for $\PCFG$s:  the equivalence problem has recently been proved \cite{fore14} to be as hard as the multiple ambiguity problem for context-free grammars: do two context-free grammars have the same number of derivation trees for each string?
Since on the other hand it is known \cite{abne99} that probabilistic pushdown automata are equivalent to  $\PCFG$s, it follows that the equivalence problem is also interreducible with the multiple ambiguity problem. Before that, the difficulty of co-emission and of related results was shown in \cite{jago01}.

Since computing distances is at least as hard as checking equivalence (the case where the distance is 0 giving us the answer to the equivalence question) it remains to be seen, for $\PCFG$s, just how hard it is to compute a distance, and also to see if all distances are as hard.

This is the subject of this work. We have studied a number of distances over distributions, including the $\Lone$ and $\Ltwo$ norms, the Hellinger and the variation distances and the Kullback-Leibler divergence. We report that none of these can be computed.

On the other hand, the fact that the consensus string can be computed allows to show that  the Chebyshev distance (or $\Linf$ norm) belongs to the same class as the multiple ambiguity problem for context-free grammars and the equivalence problem for $\PCFG$s.

In Section~\ref{sec:def} we remind the reader of the different notations, definitions and key results we will be needing. In Section~\ref{sec:res} we go through the new results we have proved. We conclude in Section~\ref{sec:con}.

\section{Definitions and notations}\label{sec:def}
Let $[n]$ denote the set $\{1,\ldots,n\}$ for each $n\in\Nrat$. Logarithms will be taken in base 2.
An \emph{alphabet} $\Sigma$ is a finite non-empty set of symbols.
A \textit{string} $w$ over $\Sigma$ is a finite sequence $w = a_1 \ldots a_n$
of symbols. Symbols  will be indicated by $a, b, c,\ldots$,
and  strings by $u, v,\ldots, z$.
Let $|w|$ denote the length of $w$.
The \emph{empty string} is denoted by $\es$.

We denote by $\Sigmastar$ the set of all strings and by $\Sigma^{\le n}$
the set of those of length at most $n$. Similarly, 
$\Sigma^{n}=\{x\in\Sigmastar:\,|x|=n\}$,
$\Sigma^{<n}=\{x\in\Sigmastar:\,|x|<n\}$ and 
$\Sigma^{\ge n}=\{x\in\Sigmastar:\,|x|\ge n\}$.

A \emph{probabilistic language} $\Dis$
is a probability distribution over $\Sigmastar$.
The probability of a string $x \in \Sigmastar$ under the
distribution $\Dis$ is denoted as
$\Prob_{\Dis}(x)$\label{index:prob}
and must satisfy $\sum_{x \in \Sigmastar} \Prob_{\Dis}(x) = 1$.
If $A$ is a language (thus a set of strings, included in $\Sigmastar$),
and $\Dis$ a distribution over $\Sigmastar$,
$\Prob_{\Dis}(A)=\sum_{x\in A} \Prob_{\Dis}(x)$.

If the distribution is modelled by some grammar $G$,
the probability of $x$ according to the probability distribution
defined by $G$ is denoted by
$\Prob_{G}(x)$. The distribution modelled by a grammar $G$ will
be denoted by ${\Dis}_{G}$.

\subsection{Context-free grammars}

A \emph{context-free grammar} is a tuple $<\Sigma, V, R,
S>$ where $\Sigma$ is a finite alphabet (of terminal symbols), $V$
is a finite alphabet (of variables or non-terminals), $R\subset V
\times (\Sigma \cup V)^*$ is a finite set of production rules, and
$S$ $(\in V)$ is the axiom or start symbol.

We will write $N\rightarrow \beta$ for rule $(N, \beta)\in R$. If
$\alpha,\beta,\gamma \in(\Sigma\cup V)^*$ and  $(N, \beta)\in R$ we have
$\alpha N\gamma\Rightarrow\alpha\beta\gamma$. This means that string $\alpha N\gamma$
\emph{derives} (in one step) into string $\alpha\beta\gamma$.

$\rgives$
is the reflexive and transitive closure of $\Rightarrow$. If there
exists $\alpha_0,\dots,\alpha_k$ such that $\alpha_0 \Rightarrow \dots \Rightarrow \alpha_k$ we will
write $\alpha_0 \stackrel{k}{\Rightarrow} \alpha_k$. $\Naming(G)$ is the language generated by $G$: the set of all strings $w$ over $\Sigma$ such that $S\rgives w$.

A sequence $\alpha_0 \Rightarrow \dots \Rightarrow \alpha_k$ is a derivation of $\alpha_k$ from $\alpha_0$. A derivation step $\alpha N\gamma\Rightarrow\alpha\beta\gamma$ is a \emph{left-most derivation step} if $\alpha\in \Sigma^*$. A derivation is \emph{left-most} if each step is left-most.

A context-free grammar is \emph{proper} if it satisfies the following three properties:
\begin{enumerate}
\item It is cycle-free, i.e. no non-terminal $A$ exists such that $A\stackrel{+}{\Longrightarrow}A$.
\item It is $\es$-free, i.e. either no rules with $\es$ on the RHS exist or exactly one exists with $S$ on the LHS (i.e. $S \rightarrow \es$) and $S$ does not appear on the RHS of any other rule.
\item It contains no useless symbols or non-terminals. This means that every symbol and non-terminal should be reachable from $S$ and every non-terminal should derive at least one string from $\Sigma^*$.
\end{enumerate}

A context-free grammar is ambiguous if there exists a string $w$ admitting two different left-most derivations from $S$ to $w$. Given any string $w$, we can define the multiplicity $m_G(w)$ as the number of different left-most derivations from $S$ to $w$. If $\forall w\in\Sigmastar\, m_G(w)\le 1$, $G$ is unambiguous. Otherwise it is ambiguous. If $\forall w\in\Sigmastar\, m_G(w)< \infty$, $G$ is a \emph{finite multiplicity grammar}. If a grammar is proper it has finite multiplicity. Two finite multiplicity grammars $G_1$ and $G_2$ are multiplicity equivalent if $\forall w\in\Sigmastar\, m_{G_1}(w)=m_{G_2}(w)$.

The multiplicity equivalence problem has been studied for many years \cite{kuic86}: the problem has been proved to be decidable only for particular classes of grammars.

Results regarding the decidability of context-free grammars can be found in many textbooks \cite{harr78}:
\begin{enumerate}
\item  Given two context-free grammars $G_1$ and $G_2$, the (\emph{equivalence}) question $\Naming(G_1)=\Naming(G_2)$? is undecidable.
\item  Given two context-free grammars $G_1$ and $G_2$, the (\emph{inclusion}) question $\Naming(G_1)\subseteq \Naming(G_2)$? is undecidable.
\item  Given two context-free grammars $G_1$ and $G_2$, the (\emph{emptiness of intersection}) question $\Naming(G_1)\cap \Naming(G_2)=\emptyset$? is undecidable.
\end{enumerate}

\subsection{Probabilistic context-free grammars}

\begin{definition}
  A \emph{probabilistic context-free grammar ($\PCFG$)} $G$
  is a context-free grammar $<\Sigma,V, R, S>$  with a probability function 
  $P:R\to\R^+$.
\end{definition}

$\Prob_{G}(x)$ is the sum of all the leftmost derivations’ probabilities of $x$, where the probability of a leftmost derivation is the product of the rule probabilities used in the derivation. A $\PCFG$ $G$ is said to be consistent if $\sum_{x\in \Sigma^*} \Prob_{G}(x) = 1$. Unless otherwise specified, any $\PCFG$ mentioned from now onwards is assumed to be consistent. 

Parsing with a $\PCFG$ is usually done by adapting the Earley or the
\textsc{Cky} algorithms \cite{jelinek92}. By straightforward variants allowing every terminal to match every input position, one can compute $\Prob_G(\Sigma^n)$, still in polynomial time in $n$. By summing $\Prob_G(\Sigma^i)$ for $i<n$ one obtains 
$\Prob_G(\Sigma^{<n})$, and $\Prob_G(\Sigma^{\geq n})$ is $1-\Prob_G(\Sigma^{<n})$. Alternatively, $\Prob_G(\Sigma^{\geq n})$ can be computed directly by variants of algorithms computing prefix probabilities \cite{jelinek91,stol95}.

We denote by $\Naming(G)$ the support language of $G$, ie the set of strings of non null probability.
The class of all $\PCFG$s over alphabet $\Sigma$ will be denoted by $\PCFGclass$.

There exists an effective procedure which, given a proper $\CFG$ $G$, builds a $\PCFG$ $G'$ such that $\forall x\in\Sigma^*, \Prob_{G'}(x)>0\iff x\in \Naming(G)$. We call this procedure \textsc{Mp} for Make Probabilistic.

One possible procedure for \textsc{Mp} is to first assign uniform probabilities to the given $\CFG$, thus obtaining a possibly inconsistent $\PCFG$ which then can be converted into a consistent $\PCFG$ using the procedure explained in \cite{gecs10}.

Let us formally define the (language) equivalence problem:
\begin{definition}
Two probabilistic grammars $G_1$ and $G_2$ are \emph{(language) equivalent} if $\forall x\in\Sigma^*,\Prob_{G_1}(x)=\Prob_{G_2}(x)$.
We denote by $\langle EQ,\PCFGclass\rangle$ the decision problem: are two $\PCFG$s $G_1$ and $G_2$ equivalent?
\end{definition}
The following result holds for  probabilistic pushdown automata, which are shown in \cite{abne99} to be equivalent to $\PCFG$s.
\begin{proposition}\cite{fore14}
The $\langle EQ,\PCFGclass\rangle$ problem is interreducible with the multiplicity equivalence problem for $\CFG$s.
\end{proposition}

\subsection{About co-emissions} 

Co-emission has been identified as a key concept allowing, in the case of hidden Markov models or probabilistic finite-state automata, computation in polynomial time of the distance for the $\Ltwo$ norm (and more generally any Lp norm, for even values of p): the distance can be computed as a finite sum of co-emissions.
\begin{definition}
The coemission of two probabilistic grammars $G_1$ and $G_2$ is the probability that both grammars $G_1$ and $G_2$ simultaneously emit the same string:\\
$\COEM(G_1,G_2)=\sum_{x\in\Sigma^*}\Prob_{G_1}(x)\cdot\Prob_{G_2}(x)$
\end{definition}
A particular case of interest is the probability of twice generating the same string when using the same grammar:
Given a $\PCFG$ $G$, the \emph{auto-coemission} of $G$, denoted as $\AC(G)$, is $\COEM(G,G)$. 
If the grammars are ambiguous, internal factorization is required in the computation of co-emission. In order to detect this we introduce the \emph{tree-auto-coemission} as the probability that the grammar produces exactly the same left-most derivation (which corresponds to a specific tree):
\begin{definition}
Given a $\PCFG$ $G$, the \emph{tree-auto-coemission} of $G$ is the probability that $G$ generates the exact same left-most derivation twice. We denote it by $\TAC(G)$.
\end{definition}
Note that the tree-auto-coemission and the auto-coemission coincide if and only if $G$ is unambiguous:
\begin{proposition}\label{ambiguity}
Let $G$ be a $\PCFG$.\\
$\AC(G)\ge\TAC(G)$.\\
$\AC(G)=\TAC(G)\iff $ $G$ is unambiguous.
\end{proposition}

\subsection{Distances between distributions}
A $\PCFG$ defines a distribution over $\Sigmastar$. If two grammars can be compared syntactically, they can also be compared semantically: do they define identical distributions, and, if not, how different are these distributions?
\begin{definition}
The $\Lone$ distance (or Manhattan distance) between two probabilistic grammars $G_1$ and $G_2$ is:
\begin{equation*}
d_{\Lone}(G_1,G_2)=\sum_{x\in\Sigma^*}|\Prob_{G_1}(x)-\Prob_{G_2}(x)|
\end{equation*}
\end{definition}
\begin{definition}
The $\Ltwo$ distance (or Euclidian distance) between two probabilistic grammars $G_1$ and $G_2$ is:
\begin{equation*}
d_{\Ltwo}(G_1,G_2)=\sqrt{\sum_{x\in\Sigma^*}(\Prob_{G_1}(x)-\Prob_{G_2}(x))^2}
\end{equation*}
\end{definition}
$\Ltwo$ distance can be rewritten in terms of coemission, as:
\begin{equation*}\label{eq:L2}
d_{\Ltwo}(G_1,G_2)=\sqrt{\COEM(G_1,G_1)-2\COEM(G_1,G_2)+\COEM(G_2,G_2)}
\end{equation*}

\begin{definition}
The $\Linf$ distance (or Chebyshev distance) between two probabilistic grammars $G_1$ and $G_2$ is:
\begin{equation*}
d_{\Linf}(G_1,G_2)=\max_{x\in\Sigma^*}|\Prob_{G_1}(x)-\Prob_{G_2}(x)|
\end{equation*}
\end{definition}
Note that the $\Linf$ distance seems closely linked with the consensus string, which is the most probable string in a language.
\begin{definition}
The \emph{variation distance} between two probabilistic grammars $G_1$ and $G_2$ is:
\begin{equation*}
d_{V}(G_1,G_2)=\max_{X\subset\Sigma^*}\sum_{x\in X}(\Prob_{G_1}(x)-\Prob_{G_2}(x))
\end{equation*}
\end{definition}
The variation distance looks like $d_{\Linf}$, but is actually connected with $d_{\Lone}$:

\begin{equation}\label{equation:variation}
d_{V}(G_1,G_2)=\frac{1}{2}d_{\Lone}(G_1,G_2)
\end{equation}
A number of distances have been studied elsewhere (for example \cite{jago01,cort07}):
\begin{definition}
The \emph{Hellinger distance} between two probabilistic grammars $G_1$ and $G_2$ is:
\begin{equation*}
d_{H}(G_1,G_2)=\frac{1}{2}\cdot\sum_{w\in\Sigma^*}\bigg(\sqrt\Prob_{G_1}(w)-\sqrt\Prob_{G_2}(w)\bigg)^2
\end{equation*}

The \emph{Jensen-Shannon (JS)} distance between two probabilistic grammars $G_1$ and $G_2$ is:
\begin{equation*}
d_{JS}(G_1,G_2)=\sum_{x \in \Sigma^*}\bigg(\Prob_{G_1}(x)\log\frac{2\Prob_{G_1}(x)}{\Prob_{G_1}(x)+\Prob_{G_2}(x)} + \Prob_{G_2}(x)\log\frac{2\Prob_{G_2}(x)}{\Prob_{G_1}(x)+\Prob_{G_2}(x)}\bigg)
\end{equation*}

The \emph{chi-squared ($\chi^2$)} distance between two probabilistic grammars $G_1$ and $G_2$ is:
\begin{equation*}
d_{\chi^2}(G_1,G_2) = \sum_{x \in \Sigma^*}\frac{(\Prob_{G_1}(x)-\Prob_{G_2}(x))^2}{\Prob_{G_2}(x)}
\end{equation*}
\end{definition}

The Kullback-Leibler divergence, or relative entropy is not a metric:
\begin{definition}
The KL divergence between two probabilistic grammars $G_1$ and $G_2$ is:
\begin{equation}
d_{KL}(G_1,G_2)=\sum_{x\in\Sigma^*}\Prob_{G_1}(x)\big(\log\Prob_{G_1}(x) - \log\Prob_{G_2}(x)\big)
\end{equation}
\end{definition}
Even if the KL-divergence does not respect the triangular inequality, $d_{KL}(G_1,G_2)=0\iff G_1\equiv G_2$.
\begin{definition}
Let $G$ be a $\PCFG$. The consensus string for $G$ is the most probable string of $\Dis_G$.
We denote by $\langle CS,\PCFGclass\rangle$ the decision problem: is $w$ the most probable string given $G$?
\end{definition}

\subsection{PCP and the probabilistic grammars}
\begin{definition}
For each distance  $d_X$ the problem $\langle d_X,\Grammarclass\rangle$ is the decision problem: given two grammars $G$ and $G'$ from $\Grammarclass$, and any rational $k$, do we have $d_X(G,G')\le k$?
\end{definition}
We  use \emph{Turing reduction} between decision problems and write:
$$
\Pi_1\le_T\Pi_2
$$
for problem $\Pi_1$ reduces to problem $\Pi_2$: if there exists a terminating algorithm solving $\Pi_2$ there also is one solving $\Pi_1$. If simultaneously $\Pi_1\le_T\Pi_2$ and $\Pi_1\le_T\Pi_2$, we will say that $\Pi_1$ and $\Pi_2$ are interreducible. The construction can be used for non decision problems: if only $\Pi_1$ is a decision problem, $\Pi_1$ is undecidable, and $\Pi_1\le_T\Pi_2$,  we will say that $\Pi_2$ is \emph{uncomputable}.

One particular well-known undecidable problem can be used as starting point for the reductions: the Post Correspondence Problem \cite{post46}, which is undecidable:

\noindent\textbf{Name:} $\PCP$  \\
\textbf{Instance:} A finite set $F$ of pairs of strings $(u_i,v_i), 1\le i\le n$ over an alphabet $\Sigma$.\\
\textbf{Question:} Is there a finite sequence of integers $x_1\dots x_t,\; t>0 $ such that $u_{x_1}u_{x_2}\dots u_{x_t}=v_{x_1}v_{x_2}\dots v_{x_t}$?
\vspace{0.1in}

We give two standard constructions starting from an instance $F$ of $\PCP$. In both cases we use another alphabet containing one symbol $\#_i$ for each $i:\, 1\le i\le n$. Let $\Omega$ denote this alphabet.

\vspace{0.1in}
\noindent\textbf{Construction 1:} Two grammars  \\
An instance of $\PCP$ as above is transformed into two $\PCFG$s $G_1$ and $G_2$ as follows:\\
Rules of $G_1$: $S_1\gives u_iS_1\#_i$ and $S_1\gives u_i\#_i$ \\
Rules of $G_2$: $S_2\gives v_iS_2\#_i$ and $S_2\gives v_i\#_i$\\
Each rule has  probability $\frac{1}{2n}$.

\vspace{0.1in}
\noindent\textbf{Construction 2:} One grammar  \\
An instance of $\PCP$ is first transformed into two $\PCFG$s $G_1$ and $G_2$ as above.
Then a new non-terminal $S_0$ is introduced and we add the new rules $S_0\gives S_1$ and $S_0\gives S_2$, each with probability $\frac{1}{2}$.

\vspace{0.1in}
The language obtained through $G_0$, $G_1$ and $G_2$ contains only strings $x$ which can always be decomposed into $x=yz$ with $y\in\Sigma^*$ and $z\in\Omega^*$. We note that the number of  derivation steps to generate string $x$ is $1+|z|$ for $G_1$ and $G_2$. For a final string $x$ we denote this number by $\len(x)$. For example $\len(aabababa\#_3\#_1\#_4\#_1)=4$.

Note that a positive instance of $\PCP$ will lead to $G_1$ and $G_2$ with a non empty intersection, and to an ambiguous $G_0$.

The following properties hold:

\begin{property}~
\begin{itemize}
\item $G_1$ and $G_2$ are unambiguous and deterministic.
\item If $x\in \Naming(G_1)$, $\Prob_{G_1}(x)=(\frac{1}{2n})^{\len(x)}$
\item $F$ is a positive instance of $\PCP$ if and only if $\COEM(G_1,G_2)>0$
\item $F$ is a positive instance of $\PCP$ if and only if $G_0$ is ambiguous. In terms of probabilities, if there exists a string $x$ such that $\Prob_{G_0}(x)=(\frac{1}{2n})^{\len(x)}$
\end{itemize}
\end{property}

\begin{property}\label{tac_for_pcp}
Let $F$ be an instance of $\PCP$ with $n$ pairs. Then
$$
\TAC(G_0)=\frac{1}{16n-4}
$$
\end{property}
\begin{eqnarray*}
\TAC(G_0)&=&\sum_{i\ge 1}\bigg(n^i\cdot\frac{1}{4}\cdot(\frac{1}{2n})^{2i}\bigg)\\
&=&\frac{n}{16n^2}\cdot\sum_{i\ge 0}(\frac{1}{4n})^i\\
&=&\frac{1}{16n}\cdot\frac{1}{1-\frac{1}{4n}}\\
&=&\frac{1}{16n}\cdot\frac{4n}{4n-1}\\
&=&\frac{1}{16n-4}
\end{eqnarray*}

\section{Some decidability results}\label{sec:res}
We report results concerning the problems related to equivalence and decision computation. 
\subsection{With the Manhattan distance}
\begin{proposition}
One cannot compute, given two $\PCFG$s $G_1$ and $G_2$,
$$
d_{\Lone}(G_1,G_2)=\sum_{x\in\Sigma^*}|\Prob_{G_1}(x)-\Prob_{G_2}(x)|.
$$
\end{proposition}
\begin{proof}
If this were possible, then we could easily solve the \emph{empty intersection problem} by building  \textsc{Mp}($G$) and \textsc{Mp}($G'$) and then checking if $d_{\Lone}(\mbox{\textsc{Mp}}(G),\mbox{\textsc{Mp}}(G'))=2$.
\end{proof}
\begin{corollary}
One cannot compute, given two $\PCFG$s $G_1$ and $G_2$, the variation distance between $G_1$ and $G_2$.
\end{corollary}
The above follows from a straightforward application of Equation~\ref{equation:variation}.

The same construction can be used for the Hellinger and Jenson-Shannon distances: whenever $\Naming(G_1)$ and $\Naming(G_2)$ have an empty intersection, it follows that $d_{H}(G_1,G_2)=1$ and $d_{JS}(G_1,G_2)=2$.

Summarizing, $\langle d_{\Lone},\PCFGclass\rangle$, $\langle d_{V},\PCFGclass\rangle$, $\langle d_{H},\PCFGclass\rangle$ and $\langle d_{JS},\PCFGclass\rangle$  are undecidable.

\subsection{With the Euclidian distance}
For $\PFA$, positive results were obtained in this case: the distance can be computed, both for $\PFA$ and $\HMM$ in polynomial time\cite{lyng99}.

In \cite{jago01}, Jagota \textit{et al.} gave the essential elements allowing a proof that co-emission, the $\Ltwo$ and the Hellinger distances are uncomputable. In order to be complete, we reconstruct similar results in this section.

\begin{proposition}
One cannot compute, given two {\PCFG}s $G_1$ and $G_2$,
$$
\COEM(G_1,G_2)=\sum_{x\in\Sigma^*}\Prob_{G_1}(x)\cdot\Prob_{G_2}(x).
$$
\end{proposition}

\begin{proof} If this were possible, then we could easily solve the empty intersection problem, by taking $G$ and $G'$, building  \textsc{Mp}($G$) and \textsc{Mp}($G'$) and then checking if $\COEM(\mbox{\textsc{Mp}}(G),\mbox{\textsc{Mp}}(G'))=0$.
\end{proof}
\begin{proposition}\label{coemd2reduction}
Computing the auto-coemission $\AC$ is at least as difficult as computing the $\Ltwo$ distance.
\end{proposition}
\begin{proof}
Suppose we have an algorithm to compute the $\Ltwo$ distance. Then given any grammar $G$, we build a dummy grammar $G_D$ which only generates, with probability 1, a single string over a different alphabet. It then follows that
$$d_{\Ltwo}(G,G_D)=\sqrt{\COEM(G,G)-2\COEM(G,G_D)+\COEM(G_D,G_D)},$$
and since the intersection between the support languages for $G$ and $G_D$ is empty, $\COEM(G,G_D)=0$. On the other hand $\COEM(G_D,G_D)=1$, trivially.
Therefore $\COEM(G,G)=d_{\Ltwo}(G,G_D)^2-1$.
\end{proof}

\begin{corollary}
One cannot compute, given two $\PCFG$s $G_1$ and $G_2$, the $\Ltwo$ distance between $G_1$ and $G_2$.
\end{corollary}
\begin{proof}
By proposition \ref{coemd2reduction} all we have to prove is that computing the auto-coemission is impossible.
Let $G_0$ be the probabilistic context-free grammar built from an instance of $\PCP$. Suppose we can compute $\AC(G_0)$. Then since (by Property \ref{tac_for_pcp}) we can compute $\TAC(G_0)$, one could then solve the ambiguity problem via Proposition \ref{ambiguity}. This is impossible.
\end{proof}
Summarizing, $\langle d_{\Ltwo},\PCFGclass\rangle$ and $\langle \COEM,\PCFGclass\rangle$ are undecidable. Furthermore $\langle \AC,\PCFGclass\rangle$ and $\langle d_{\Ltwo},\PCFGclass\rangle$ are interreducible.

\subsection{With the KL divergence}
\begin{proposition}
One cannot compute, given two $\PCFG$s $G_1$ and $G_2$, $d_{KL}(G_1,G_2)$.
\end{proposition}
\begin{proof}
Suppose we could compute the KL divergence between two $\PCFG$s. We should note that $d_{KL}(G_1,G_2)<\infty$ if and only if $\Naming(G_1)\subseteq \Naming(G_2)$. We would therefore be able to check if one context-free language is included in another one, which we know is impossible.
\end{proof}
The same proof can be used for the $\chi^2$ distance since $d_{\chi^2}(G_1,G_2)<\infty$ if and only if $\Naming(G_1)\subseteq \Naming(G_2)$. Summarizing, $\langle d_{KL},\PCFGclass\rangle$ and $\langle d_{\chi^2},\PCFGclass\rangle$ are undecidable.

\subsection{About the consensus string}
A first result is that computing the Chebyshev distance is at least as difficult as computing the most probable string :
Any $\PCFG$ $G$ can be converted into $G'$ with the same rules as $G$ but using a disjoint alphabet. Now, it is clear that $d_{\Linf}(G,G')=\Prob_G(CS(G))$:
\begin{proposition}\label{linking_CS_Linf}
$\langle CS,\PCFGclass\rangle$ is decidable if there exists an algorithm computing $\Linf$. \\
Ie, $\langle CS,\PCFGclass\rangle\le_T\langle d_{\Linf},\PCFGclass\rangle$.
\end{proposition}
Proposition \ref{linking_CS_Linf} does not preclude that $\langle CS,\PCFGclass\rangle$ may be decidable if $\langle \Linf,\PCFGclass\rangle$ is not.\

In fact $\langle CS,\PCFGclass\rangle$ is decidable:
\begin{lemma}\label{yes we can}
Let $0<\epsilon<1$. Given any consistent $\PCFG$ $G$, there exists $n\ge0$ such that $\Prob_G(\Sigma^{>n})<\epsilon$
\end{lemma}

\begin{proof}
Because $\lim_{n\rightarrow +\infty} 
\Prob(\Sigma^{\le n}) = 1$, 
there exists $n$ such that 
$\Prob(\Sigma^{\le n}) \geq 1 - \epsilon$,
hence $\Prob(\Sigma^{>n})<\epsilon$.
\end{proof}

\begin{proposition}\label{CS_is_decidable}
$\langle CS,\PCFGclass\rangle$ is decidable.
\end{proposition}
\begin{proof}
The proof for this is the existence of an algorithm that takes as input a PCFG $G$ and always terminates by returning the consensus string for $G$. Algorithm \ref{algo:enumeration} does exactly this.

\end{proof}

Algorithm \ref{algo:enumeration} goes through all the possible strings in $\Sigma^0$, $\Sigma^1$, $\Sigma^2\ldots$, and checks the probability that $G$ assigns to each string. It stores the string with the highest probability value ($Current\_Best$) and the highest probability value itself ($Current\_Prob$). It also subtracts from 1 all the probability values encountered ($Remaining\_Prob$). So, after the $i^{th}$ loop, $Current\_Best$ is the most probable string in $\Sigma^{<i}$, $Current\_Prob$ is the probability of $Current\_Best$ and $Remaining\_Prob$ is $1 - \Prob(\Sigma^{< i})$ which is equal to $\Prob(\Sigma^{\geq i})$. Using Lemma \ref{yes we can}, we can say that for any $\epsilon$, $0 < \epsilon < 1$, there exists an $i$ such that after the $i^{th}$ iteration, $Remaining\_Prob$ is smaller than $\epsilon$. This means that the algorithm must halt at some point. Moreover, if the most probable string in $\Sigma^{<i}$ has probability higher than $\Prob(\Sigma^{\geq i})$, then we can be sure that this is the consensus string. This means that the algorithm always returns the consensus string.

\begin{algorithm}

  \SetKw{Kwtrue}{true}
  \SetKw{Kwfalse}{false}
  \SetKw{Kwand}{and}
  \KwData{ a $\PCFG$ $G$}
  \KwResult{$w$, the most probable string}
$Current\_Prob=0$\;
$Current\_Best=\lambda$\;
$Remaining\_Prob=1$\;
$n=0$\;
$Continue=\Kwtrue$\;
\While{$Continue$}
{  \If{$Remaining\_Prob < Current\_Prob$}
   {$Continue=\Kwfalse$\;}
   \Else
   {\ForEach{ $w \in \Sigma^n$ }
	{ $p = \Prob_{G}(w)$\;
      $Remaining\_Prob = Remaining\_Prob - p$\;
	  \If{$p>Current\_Prob$}
	    {$Current\_Prob=p$\;$Current\_Best=w$\;}
    }
    $n=n+1$\;
   }	  
}
\Return $Current\_Best$
\caption{Finding the consensus string}\label{algo:enumeration}
\end{algorithm}

\subsection{With the Chebyshev distance}
\begin{proposition}
$\langle d_{\Linf},\PCFGclass\rangle$ and $\langle EQ,\PCFGclass\rangle$ are interreducible.
\end{proposition}
\begin{proof}
$\langle d_{\Linf},\PCFGclass\rangle\le_T \langle EQ,\PCFGclass\rangle.$\\
Suppose $\langle EQ,\PCFGclass\rangle$ is decidable. Then if $G_1$ and $G_2$ are equivalent, $d_{\Linf}(G_1,G_2)=0$.
If $G_1$ and $G_2$ are not equivalent, there exists a smallest string $x$ such that $\Prob_{G_1}(x)\neq\Prob_{G_2}(x)$. An enumeration algorithm will find this initial string $x$, whose probability is $p$. Note that if $\Prob_{G_1}(\Sigma^{>n})<p$ and $\Prob_{G_2}(\Sigma^{>n})<p$, we can be sure that no string in $\Sigma^{>n}$ has a difference of probabilities of more than $p$. This allows us to adapt Algorithm \ref{algo:enumeration} to reach the length $n$ at which we are sure that no string $x$ longer than $n$ can have probability more than $|\Prob_{G_1}(x)-\Prob_{G_2}(x)|$. The algorithm will therefore halt.

The converse ($\langle EQ,\Grammarclass\rangle\le_T\langle d_{\Linf},\Grammarclass\rangle$) is trivial since $d_{\Linf}(G_1,G_2)=0\Leftrightarrow$ $G_1$ and $G_2$ are equivalent.
\end{proof}

\section{Conclusion}\label{sec:con}
The results presented in this work are threefold:
\begin{itemize}
\item the only positive result concerns the consensus string, which is computable;
\item the multiple ambiguity problem (for $\CFG$s) is equivalent to the Chebyshev distance problem (for $\PCFG$s), which in turn is equivalent to the equivalence problem (also for $\PCFG$s);
\item all the other results correspond to undecidable problems.
\end{itemize}

Interestingly, if we consider the Chebyshev distance problem as a decision problem, namely:

\vspace{0.1in}
\noindent\textbf{Name:} Chebyshev distance-$\le$  \\
\textbf{Instance:} Two $\PCFG$s $G_1$ and $G_2$, $\epsilon:\,0\le\epsilon\le 1$\\
\textbf{Question:} Is $d_{\Linf}(G_1,G_2)\le \epsilon$?

\vspace{0.1in}
the problem is actually decidable in all cases but one: when $\epsilon=0$.
Ideally, one would hope to be able to bound the length of the strings over which the search should be done. This is possible in the case of probabilistic finite automata where it is proved that (1) a $\PFA$ can be transformed into an equivalent $\es$-free $\PFA$ $\A$, and, (2) the length of any string of probability at least $p$ is upper bounded by $\frac{|\A|}{p^2}$, with $|\A|$ the size (number of states+1) of the $\PFA$ \cite{higu13b}.

It should be noted that if the question is: Is $d_{\Linf}(G_1,G_2)< \epsilon$?, the problem becomes decidable.

Moreover, it would be important to obtain approximation results, ie,

\vspace{0.1in}
\noindent\textbf{Name:} X-distance-approx  \\
\textbf{Instance:} Two $\PCFG$s $G_1$ and $G_2$, $\epsilon:\,0<\epsilon\le 1$\\
\textbf{Question:} Compute $a$ such that $|a-d_X(G_1,G_2)|\le \epsilon$?

\vspace{0.1in}

Such results have been studied in the case of probabilistic finite state machines, for example, recently, in \cite{chen14}. In the case of the distances used in this work, the decidability of approximation would be ensured by  Lemma \ref{yes we can}. But the question of finding good approximations in polynomial time is clearly an interesting problem.
\section*{Acknowledgement}
The authors thank James Worrell, Achilles Beros and Uli Fahrenberg for advice and discussions.
\bibliographystyle{eptcs}
\bibliography{PCFG_distances}

\begin{thebibliography}{10}
\providecommand{\bibitemdeclare}[2]{}
\providecommand{\surnamestart}{}
\providecommand{\surnameend}{}
\providecommand{\urlprefix}{Available at }
\providecommand{\url}[1]{\texttt{#1}}
\providecommand{\href}[2]{\texttt{#2}}
\providecommand{\urlalt}[2]{\href{#1}{#2}}
\providecommand{\doi}[1]{doi:\urlalt{http://dx.doi.org/#1}{#1}}
\providecommand{\bibinfo}[2]{#2}

\bibitemdeclare{inproceedings}{abne99}
\bibitem{abne99}
\bibinfo{author}{S.~\surnamestart Abney\surnameend},
  \bibinfo{author}{D.~\surnamestart McAllester\surnameend} \&
  \bibinfo{author}{F.~\surnamestart Pereira\surnameend} (\bibinfo{year}{1999}):
  \emph{\bibinfo{title}{Relating Probabilistic Grammars and Automata}}.
\newblock In: {\sl \bibinfo{booktitle}{Proceedings of the 27th Annual Meeting
  of the Association for Computational Linguistics}}.

\bibitemdeclare{mastersthesis}{bala93}
\bibitem{bala93}
\bibinfo{author}{V.~\surnamestart Balasubramanian\surnameend}
  (\bibinfo{year}{1993}): \emph{\bibinfo{title}{Equivalence and Reduction of
  Hidden {M}arkov Models}}.
\newblock Master's thesis, \bibinfo{school}{Department of Electrical
  Engineering and Computer Science, \textsc{Mit}}.
\newblock \bibinfo{note}{Issued as AI Technical Report 1370}.

\bibitemdeclare{article}{benedi05}
\bibitem{benedi05}
\bibinfo{author}{J-M. \surnamestart Bened\'{\i}\surnameend} \&
  \bibinfo{author}{J.-A. \surnamestart S{\'a}nchez\surnameend}
  (\bibinfo{year}{2005}): \emph{\bibinfo{title}{Estimation of stochastic
  context-free grammars and their use as language models}}.
\newblock {\sl \bibinfo{journal}{Computer Speech {\&} Language}}
  \bibinfo{volume}{19}(\bibinfo{number}{3}), pp. \bibinfo{pages}{249--274}.

\bibitemdeclare{article}{carr97}
\bibitem{carr97}
\bibinfo{author}{R.~C. \surnamestart Carrasco\surnameend}
  (\bibinfo{year}{1997}): \emph{\bibinfo{title}{Accurate computation of the
  relative entropy between stochastic regular grammars}}.
\newblock {\sl \bibinfo{journal}{\textsc{Rairo} (Theoretical Informatics and
  Applications)}} \bibinfo{volume}{31}(\bibinfo{number}{5}), pp.
  \bibinfo{pages}{437--444}.

\bibitemdeclare{inproceedings}{casa00a}
\bibitem{casa00a}
\bibinfo{author}{F.~\surnamestart Casacuberta\surnameend} \&
  \bibinfo{author}{C.~\surnamestart de~la Higuera\surnameend}
  (\bibinfo{year}{2000}): \emph{\bibinfo{title}{Computational complexity of
  problems on probabilistic grammars and transducers}}.
\newblock In \bibinfo{editor}{A.~L. \surnamestart de~Oliveira\surnameend},
  editor: {\sl \bibinfo{booktitle}{Grammatical Inference: Algorithms and
  Applications, Proceedings of \textsc{Icgi} '00}}, {\sl
  \bibinfo{series}{\textsc{Lnai}}} \bibinfo{volume}{1891},
  \bibinfo{publisher}{Springer-Verlag}, pp. \bibinfo{pages}{15--24}.

\bibitemdeclare{inproceedings}{chen14}
\bibitem{chen14}
\bibinfo{author}{T.~\surnamestart Chen\surnameend} \&
  \bibinfo{author}{S.~\surnamestart Kiefer\surnameend} (\bibinfo{year}{2014}):
  \emph{\bibinfo{title}{On the total variation distance of labelled {M}arkov
  chains}}.
\newblock In: {\sl \bibinfo{booktitle}{Proceedings of LICS 2014}},
  \doi{http://arxiv.org/abs/1405.2852}.

\bibitemdeclare{article}{cort07}
\bibitem{cort07}
\bibinfo{author}{C.~\surnamestart Cortes\surnameend},
  \bibinfo{author}{M.~\surnamestart Mohri\surnameend} \&
  \bibinfo{author}{A.~\surnamestart Rastogi\surnameend} (\bibinfo{year}{2007}):
  \emph{\bibinfo{title}{$L_p$ Distance and equivalence of probabilistic
  automata}}.
\newblock {\sl \bibinfo{journal}{International Journal of Foundations of
  Computer Science}} \bibinfo{volume}{18}(\bibinfo{number}{4}), pp.
  \bibinfo{pages}{761--779}.

\bibitemdeclare{article}{cort08}
\bibitem{cort08}
\bibinfo{author}{C.~\surnamestart Cortes\surnameend},
  \bibinfo{author}{M.~\surnamestart Mohri\surnameend},
  \bibinfo{author}{A.~\surnamestart Rastogi\surnameend} \&
  \bibinfo{author}{M.~\surnamestart Riley\surnameend} (\bibinfo{year}{2008}):
  \emph{\bibinfo{title}{On the Computation of the Relative Entropy of
  Probabilistic Automata}}.
\newblock {\sl \bibinfo{journal}{International Journal on Foundations of
  Computer Science}} \bibinfo{volume}{19}(\bibinfo{number}{1}), pp.
  \bibinfo{pages}{219--242}.

\bibitemdeclare{inproceedings}{espa04}
\bibitem{espa04}
\bibinfo{author}{J.~\surnamestart Esparza\surnameend},
  \bibinfo{author}{A.~\surnamestart Kucera\surnameend} \&
  \bibinfo{author}{R.~\surnamestart Mayr\surnameend} (\bibinfo{year}{2004}):
  \emph{\bibinfo{title}{Model Checking Probabilistic Pushdown Automata}}.
\newblock In: {\sl \bibinfo{booktitle}{Proceedings of LICS}},
  \bibinfo{publisher}{IEEE Computer Society}, pp. \bibinfo{pages}{12--21}.

\bibitemdeclare{article}{etessami2009}
\bibitem{etessami2009}
\bibinfo{author}{K.~\surnamestart Etessami\surnameend} \&
  \bibinfo{author}{M.~\surnamestart Yannakakis\surnameend}
  (\bibinfo{year}{2009}): \emph{\bibinfo{title}{Recursive {M}arkov chains,
  stochastic grammars, and monotone systems of nonlinear equations}}.
\newblock {\sl \bibinfo{journal}{Journal of the ACM}}
  \bibinfo{volume}{56}(\bibinfo{number}{1}), pp. \bibinfo{pages}{1--66}.

\bibitemdeclare{article}{fore12}
\bibitem{fore12}
\bibinfo{author}{V.~\surnamestart Forejt\surnameend},
  \bibinfo{author}{P.~\surnamestart Jancar\surnameend},
  \bibinfo{author}{S.~\surnamestart Kiefer\surnameend} \&
  \bibinfo{author}{J.~\surnamestart Worrell\surnameend} (\bibinfo{year}{2012}):
  \emph{\bibinfo{title}{Bisimilarity of Probabilistic Pushdown Automata}}.
\newblock {\sl \bibinfo{journal}{CoRR}} \bibinfo{volume}{abs/1210.2273}.

\bibitemdeclare{article}{fore14}
\bibitem{fore14}
\bibinfo{author}{V.~\surnamestart Forejt\surnameend},
  \bibinfo{author}{P.~\surnamestart Jancar\surnameend},
  \bibinfo{author}{S.~\surnamestart Kiefer\surnameend} \&
  \bibinfo{author}{J.~\surnamestart Worrell\surnameend} (\bibinfo{year}{2014}):
  \emph{\bibinfo{title}{Language Equivalence of Probabilistic Pushdown
  Automata}}.
\newblock {\sl \bibinfo{journal}{Information and Computation}}
  \bibinfo{volume}{237}, pp. \bibinfo{pages}{1--11}.

\bibitemdeclare{article}{fubo75}
\bibitem{fubo75}
\bibinfo{author}{K.~S. \surnamestart Fu\surnameend} \& \bibinfo{author}{T.~L.
  \surnamestart Booth\surnameend} (\bibinfo{year}{1975}):
  \emph{\bibinfo{title}{Grammatical Inference: Introduction and Survey. {P}art
  {I} and {II}}}.
\newblock {\sl \bibinfo{journal}{\textsc{Ieee} Transactions on Syst. Man. and
  Cybern.}} \bibinfo{volume}{5}, pp. \bibinfo{pages}{59--72 and 409--423}.

\bibitemdeclare{article}{gecs10}
\bibitem{gecs10}
\bibinfo{author}{R.~\surnamestart Gecse\surnameend} \&
  \bibinfo{author}{A.~\surnamestart Kovacs\surnameend} (\bibinfo{year}{2010}):
  \emph{\bibinfo{title}{Consistency of stochastic context-free grammars}}.
\newblock {\sl \bibinfo{journal}{Mathematical and Computer Modelling}}
  \bibinfo{volume}{52}(\bibinfo{number}{3–4}), pp. \bibinfo{pages}{490 --
  500}.

\bibitemdeclare{book}{harr78}
\bibitem{harr78}
\bibinfo{author}{M.~H. \surnamestart Harrison\surnameend}
  (\bibinfo{year}{1978}): \emph{\bibinfo{title}{Introduction to Formal Language
  Theory}}.
\newblock \bibinfo{publisher}{Addison-Wesley Publishing Company, Inc.},
  \bibinfo{address}{Reading, MA}.

\bibitemdeclare{book}{higu10}
\bibitem{higu10}
\bibinfo{author}{C.~\surnamestart de~la Higuera\surnameend}
  (\bibinfo{year}{2010}): \emph{\bibinfo{title}{Grammatical inference: learning
  automata and grammars}}.
\newblock \bibinfo{publisher}{Cambridge University Press}.

\bibitemdeclare{inproceedings}{higu13b}
\bibitem{higu13b}
\bibinfo{author}{C.~\surnamestart de~la Higuera\surnameend} \&
  \bibinfo{author}{J.~\surnamestart Oncina\surnameend} (\bibinfo{year}{2013}):
  \emph{\bibinfo{title}{Computing the Most Probable String with a Probabilistic
  Finite State Machine}}.
\newblock In: {\sl \bibinfo{booktitle}{Proceedings of \textsc{Fsmnlp}}}.

\bibitemdeclare{article}{higu13a}
\bibitem{higu13a}
\bibinfo{author}{C.~\surnamestart de~la Higuera\surnameend} \&
  \bibinfo{author}{J.~\surnamestart Oncina\surnameend} (\bibinfo{year}{2013}):
  \emph{\bibinfo{title}{The Most Probable String: an Algorithmic Study}}.
\newblock {\sl \bibinfo{journal}{Journal of Logic and Computation, \texttt{doi:
  10.1093/logcom/exs049}}}.

\bibitemdeclare{inproceedings}{jago01}
\bibitem{jago01}
\bibinfo{author}{A.~\surnamestart Jagota\surnameend}, \bibinfo{author}{R.~B.
  \surnamestart Lyngs{\o}\surnameend} \& \bibinfo{author}{C.~N.~S.
  \surnamestart Pedersen\surnameend} (\bibinfo{year}{2001}):
  \emph{\bibinfo{title}{Comparing a Hidden {M}arkov Model and a Stochastic
  Context-Free Grammar}}.
\newblock In: {\sl \bibinfo{booktitle}{Proceedings of \textsc{Wabi} '01}}, {\sl
  \bibinfo{series}{\textsc{Lncs}}} \bibinfo{volume}{2149},
  \bibinfo{publisher}{Springer-Verlag}, pp. \bibinfo{pages}{69--84}.

\bibitemdeclare{article}{jelinek91}
\bibitem{jelinek91}
\bibinfo{author}{F.~\surnamestart Jelinek\surnameend} \& \bibinfo{author}{J.D.
  \surnamestart Lafferty\surnameend} (\bibinfo{year}{1991}):
  \emph{\bibinfo{title}{Computation of the Probability of Initial Substring
  Generation by Stochastic Context-Free Grammars}}.
\newblock {\sl \bibinfo{journal}{Computational Linguistics}}
  \bibinfo{volume}{17}(\bibinfo{number}{3}), pp. \bibinfo{pages}{315--323}.

\bibitemdeclare{incollection}{jelinek92}
\bibitem{jelinek92}
\bibinfo{author}{F.~\surnamestart Jelinek\surnameend}, \bibinfo{author}{J.D.
  \surnamestart Lafferty\surnameend} \& \bibinfo{author}{R.L. \surnamestart
  Mercer\surnameend} (\bibinfo{year}{1992}): \emph{\bibinfo{title}{Basic
  Methods of Probabilistic Context Free Grammars}}.
\newblock In \bibinfo{editor}{P.~\surnamestart Laface\surnameend} \&
  \bibinfo{editor}{R.~\surnamestart De~Mori\surnameend}, editors: {\sl
  \bibinfo{booktitle}{Speech Recognition and Understanding --- Recent Advances,
  Trends and Applications}}, \bibinfo{publisher}{Springer-Verlag}, pp.
  \bibinfo{pages}{345--360}.

\bibitemdeclare{article}{johnson98}
\bibitem{johnson98}
\bibinfo{author}{M.~\surnamestart Johnson\surnameend} (\bibinfo{year}{1998}):
  \emph{\bibinfo{title}{PCFG Models of Linguistic Tree Representations}}.
\newblock {\sl \bibinfo{journal}{Comput. Linguist.}}
  \bibinfo{volume}{24}(\bibinfo{number}{4}), pp. \bibinfo{pages}{613--632}.

\bibitemdeclare{inproceedings}{jurafsky95}
\bibitem{jurafsky95}
\bibinfo{author}{D.~\surnamestart Jurafsky\surnameend},
  \bibinfo{author}{C.~\surnamestart Wooters\surnameend},
  \bibinfo{author}{J.~\surnamestart Segal\surnameend},
  \bibinfo{author}{A.~\surnamestart Stolcke\surnameend},
  \bibinfo{author}{E.~\surnamestart Fosler\surnameend},
  \bibinfo{author}{G.~\surnamestart Tajchaman\surnameend} \&
  \bibinfo{author}{N.~\surnamestart Morgan\surnameend} (\bibinfo{year}{1995}):
  \emph{\bibinfo{title}{Using a stochastic context-free grammar as a language
  model for speech recognition}}.
\newblock In: {\sl \bibinfo{booktitle}{Acoustics, Speech, and Signal
  Processing, 1995. ICASSP-95., 1995 International Conference on}},
  \bibinfo{volume}{1}, \bibinfo{organization}{IEEE}, pp.
  \bibinfo{pages}{189--192}.

\bibitemdeclare{inproceedings}{klein03}
\bibitem{klein03}
\bibinfo{author}{D.~\surnamestart Klein\surnameend} \&
  \bibinfo{author}{C.~\surnamestart Manning\surnameend} (\bibinfo{year}{2003}):
  \emph{\bibinfo{title}{Accurate Unlexicalized Parsing}}.
\newblock In: {\sl \bibinfo{booktitle}{Proceedings of the 41st Annual Meeting
  on Association for Computational Linguistics - Volume 1}},
  \bibinfo{series}{ACL '03}, pp. \bibinfo{pages}{423--430}.

\bibitemdeclare{book}{kuic86}
\bibitem{kuic86}
\bibinfo{author}{W.~\surnamestart Kuich\surnameend} \&
  \bibinfo{author}{A.~\surnamestart Salomaa\surnameend} (\bibinfo{year}{1986}):
  \emph{\bibinfo{title}{Semirings, Automata, Languages}}.
\newblock \bibinfo{publisher}{Springer-Verlag}.

\bibitemdeclare{inproceedings}{lyng01}
\bibitem{lyng01}
\bibinfo{author}{R.~B. \surnamestart Lyngs{\o}\surnameend} \&
  \bibinfo{author}{C.~N.~S. \surnamestart Pedersen\surnameend}
  (\bibinfo{year}{2001}): \emph{\bibinfo{title}{Complexity of Comparing Hidden
  {M}arkov Models}}.
\newblock In: {\sl \bibinfo{booktitle}{Proceedings of \textsc{Isaac} '01}},
  {\sl \bibinfo{series}{\textsc{Lncs}}} \bibinfo{volume}{2223},
  \bibinfo{publisher}{Springer-Verlag}, pp. \bibinfo{pages}{416--428}.

\bibitemdeclare{article}{lyng02}
\bibitem{lyng02}
\bibinfo{author}{R.~B. \surnamestart Lyngs{\o}\surnameend} \&
  \bibinfo{author}{C.~N.~S. \surnamestart Pedersen\surnameend}
  (\bibinfo{year}{2002}): \emph{\bibinfo{title}{The Consensus String Problem
  and the Complexity of Comparing Hidden Markov Models}}.
\newblock {\sl \bibinfo{journal}{Journal of Computing and System Science}}
  \bibinfo{volume}{65}(\bibinfo{number}{3}), pp. \bibinfo{pages}{545--569}.

\bibitemdeclare{inproceedings}{lyng99}
\bibitem{lyng99}
\bibinfo{author}{R.~B. \surnamestart Lyngs{\o}\surnameend},
  \bibinfo{author}{C.~N.~S. \surnamestart Pedersen\surnameend} \&
  \bibinfo{author}{H.~\surnamestart Nielsen\surnameend} (\bibinfo{year}{1999}):
  \emph{\bibinfo{title}{Metrics and similarity measures for hidden {M}arkov
  models}}.
\newblock In: {\sl \bibinfo{booktitle}{Proceedings of \textsc{Ismb} '99}}, pp.
  \bibinfo{pages}{178--186}.

\bibitemdeclare{inproceedings}{murg04}
\bibitem{murg04}
\bibinfo{author}{T.~\surnamestart Murgue\surnameend} \&
  \bibinfo{author}{C.~\surnamestart de~la Higuera\surnameend}
  (\bibinfo{year}{2004}): \emph{\bibinfo{title}{Distances between
  Distributions: Comparing Language Models}}.
\newblock In \bibinfo{editor}{A.~\surnamestart Fred\surnameend},
  \bibinfo{editor}{T.~\surnamestart Caelli\surnameend},
  \bibinfo{editor}{R.~\surnamestart Duin\surnameend},
  \bibinfo{editor}{A.~\surnamestart Campilho\surnameend} \&
  \bibinfo{editor}{D.~\surnamestart de~Ridder\surnameend}, editors: {\sl
  \bibinfo{booktitle}{Structural, Syntactic and Statistical Pattern
  Recognition, Proceedings of \textsc{Sspr} and \textsc{Spr} 2004}}, {\sl
  \bibinfo{series}{\textsc{Lncs}}} \bibinfo{volume}{3138},
  \bibinfo{publisher}{Springer-Verlag}, pp. \bibinfo{pages}{269--277}.

\bibitemdeclare{inproceedings}{nede04}
\bibitem{nede04}
\bibinfo{author}{M-J. \surnamestart Nederhof\surnameend} \&
  \bibinfo{author}{G.~\surnamestart Satta\surnameend} (\bibinfo{year}{2004}):
  \emph{\bibinfo{title}{Kullback-Leibler distance between probabilistic
  context-free grammars and probabilistic finite automata}}.
\newblock In: {\sl \bibinfo{booktitle}{Proceedings of \textsc{Coling} '04
  Proceedings of the 20th international conference on Computational
  Linguistics}}, \bibinfo{volume}{71}.

\bibitemdeclare{article}{post46}
\bibitem{post46}
\bibinfo{author}{E.~L. \surnamestart Post\surnameend} (\bibinfo{year}{1946}):
  \emph{\bibinfo{title}{A variant of a recursively unsolvable problem}}.
\newblock {\sl \bibinfo{journal}{Bulletin of the American Mathematical
  Society}} \bibinfo{volume}{52}(\bibinfo{number}{4}), pp.
  \bibinfo{pages}{264--268}.

\bibitemdeclare{article}{saka94}
\bibitem{saka94}
\bibinfo{author}{Y.~\surnamestart Sakakibara\surnameend},
  \bibinfo{author}{M.~\surnamestart Brown\surnameend},
  \bibinfo{author}{R.~\surnamestart Hughley\surnameend},
  \bibinfo{author}{I.~\surnamestart Mian\surnameend},
  \bibinfo{author}{K.~\surnamestart Sjolander\surnameend},
  \bibinfo{author}{R.~\surnamestart Underwood\surnameend} \&
  \bibinfo{author}{D.~\surnamestart Haussler\surnameend}
  (\bibinfo{year}{1994}): \emph{\bibinfo{title}{Stochastic context-free
  grammars for {tRNA} modeling}}.
\newblock {\sl \bibinfo{journal}{Nuclear Acids Research}} \bibinfo{volume}{22},
  pp. \bibinfo{pages}{5112--5120}.

\bibitemdeclare{article}{salv02}
\bibitem{salv02}
\bibinfo{author}{I.~\surnamestart Salvador\surnameend} \& \bibinfo{author}{J-M.
  \surnamestart Bened\'i\surnameend} (\bibinfo{year}{2002}):
  \emph{\bibinfo{title}{\textsc{Rna} Modeling by Combining Stochastic
  Context-Free Grammars and $n$-Gram Models}}.
\newblock {\sl \bibinfo{journal}{International Journal of Pattern Recognition
  and Artificial Intelligence}} \bibinfo{volume}{16}(\bibinfo{number}{3}), pp.
  \bibinfo{pages}{309--316}.

\bibitemdeclare{article}{stol95}
\bibitem{stol95}
\bibinfo{author}{A.~\surnamestart Stolcke\surnameend} (\bibinfo{year}{1995}):
  \emph{\bibinfo{title}{An Efficient probablistic Context-Free Parsing
  Algorithm that Computes Prefix Probabilities}}.
\newblock {\sl \bibinfo{journal}{Computational Linguistics}}
  \bibinfo{volume}{21}(\bibinfo{number}{2}), pp. \bibinfo{pages}{165--201}.

\end{thebibliography}
\end{document}